\documentclass[a4paper,12pt]{article}
\usepackage{amsmath,amssymb,amsfonts,amsthm}
\newtheorem{theorem}{Theorem}
\newtheorem{lemma}{Lemma}

\usepackage{graphicx}

\let\oldsqrt\sqrt
\def\sqrt{\mathpalette\DHLhksqrt}
\def\DHLhksqrt#1#2{%
\setbox0=\hbox{$#1\oldsqrt{#2\,}$}\dimen0=\ht0
\advance\dimen0-0.2\ht0
\setbox2=\hbox{\vrule height\ht0 depth -\dimen0}%
{\box0\lower0.4pt\box2}}

\renewcommand{\r}{\rho}
\renewcommand{\=}{\; \dot{=} \;}
\newcommand{\p}{\partial}
\newcommand{\A}{\mathcal{A}}
\newcommand{\E}{\mathcal{E}}

\newcommand{\Rt}{\mathbb{R}^3}
\newcommand{\sph}{\mathbb{S}^2}
\newcommand{\dv}{\, dv}

\newcommand{\Lsp}{\Delta}
\newcommand{\Hi}{\mathcal{H}}
\newcommand{\C}{S_{t}}

\newcommand{\D}{\mathcal{D}}

\title{The wave equation on the extreme Reissner-Nordstr\"om black hole}

\author{Sergio Dain$^{1,2}$ and  Gustavo Dotti$^{1}$\\ 
  \\
  $^1$Facultad de Matem\'atica, Astronom\'{i}a y F\'{i}sica, FaMAF,\\
  Universidad Nacional de C\'ordoba,\\
  Instituto de F\'{\i}sica Enrique Gaviola, IFEG, CONICET,\\
  Ciudad Universitaria, (5000) C\'ordoba, Argentina.  \\
  $^{2}$Max Planck Institute for Gravitational Physics,\\
  (Albert Einstein Institute), Am M\"uhlenberg 1,\\
  D-14476 Potsdam Germany. }

\begin{document}
\maketitle

\begin{abstract}
  We study the scalar wave equation on the open exterior region of an extreme
  Reissner-Nordstr\"om black hole and prove that, given compactly supported
  data on a Cauchy surface orthogonal to the timelike Killing vector field, the
  solution, together with its $(t,s,\theta,\phi)$ derivatives of arbitrary
  order, $s$ a tortoise radial coordinate, is bounded by a constant that
  depends only on the initial data. Our technique does not allow to study
  transverse derivatives at the horizon, which is outside the coordinate patch
  that we use.  However, using previous results that show that second and
  higher transverse derivatives at the horizon of a generic solution grow
  unbounded along horizon generators, we show that any such a divergence, if
  present, would be milder for solutions with compact initial data.
\end{abstract}

\section{Introduction}
\label{sec:introduction}

Extreme black holes lie in the boundary between black holes and naked
singularities. Since black holes are believed to be astrophysically relevant,
whereas naked singularities are considered unphysical, the issue of stability
of extreme black holes is key in understanding the process of gravitational
collapse.  Black hole stability is a longstanding open problem in General
Relativity.  The pioneering works of Regge, Wheeler \cite{Regge:1957td},
Zerilli \cite{Zerilli:1970se} \cite{Zerilli:1974ai} and Moncrief
\cite{Moncrief:1975sb} determined the modal linear stability of
electro-gravitational perturbations in the domain of outer communication of the
spherically symmetric electro-vacuum black holes, by ruling out exponential
growth in time.  Since then, a lot of effort has been made to establish more
accurate bounds on linear fields.  In particular, analyzing the scalar wave
equation on the black hole background provides useful insight into the more
complex problem of linear gravitational perturbations. Kay and Wald
(\cite{wald:1056} \cite{wald:218} \cite{Kay:1987ax}) obtained uniform
boundedness for solutions of the wave equation on the exterior of the
Schwarzschild black hole. In recent years this result has been extended to the
non-extreme Kerr black hole (see the review articles \cite{Dafermos:2008en},
\cite{Dafermos:2010hd} and references therein.)  The purpose of this work is to
place pointwise bounds on scalar waves on the exterior region of an extreme
Reissner-Nordstr\"om black hole.  The interest in Reissner-Nordstr\"om black
holes lies in the fact that they share the complexity of the global structure
of the more relevant Kerr black holes and, due to spherical symmetry, are far
more tractable than the rotating holes. The modal stability of the outer region
of Reissner-Nordstr\"om black holes under linear perturbations of the metric
and electromagnetic fields was established in \cite{Zerilli:1970se}
\cite{Zerilli:1974ai} \cite{Moncrief:1975sb}, both for the extreme and
sub-extreme cases. The modal {\em instability} of the Reissner-Nordstr\"om
naked singularity, and also of the black hole inner static region, was proved
only recently in
\cite{Dotti:2010uc}. 

The wave equation on extreme Reissner-Nordstr\"om black holes has recently been
studied by Aretakis in a series of relevant articles
\cite{Aretakis:2011ha,Aretakis:2011hc,Aretakis:2010gd}, where it was found that
second and higher order transverse derivatives at the horizon grow without
bound along the horizon generators (see also \cite{Lucietti:2012sf} were
similar results were found for the Teukolsky equation on an extreme Kerr black
hole). One of the motivations of our article is understanding the meaning of
these instabilities. More specifically, we wonder if they arise in the
evolution of fields from data of compact support on a $t=$ constant Cauchy
surface (``compact data", for short), which is a subclass in
\cite{Aretakis:2010gd,Aretakis:2011hc,Aretakis:2011ha} for which, as we show
below, the proof of instability there fails.  Although we do not prove that the
horizon instability is absent for compact data, we do show that, if present, is
milder. On the other side, we get a remarkably simple proof of pointwise
boundedness of fields of compact data and their partial derivatives of any
order in $(t,s,\theta,\phi)$ coordinates on the open black hole exterior
region. These results are stated under Theorem \ref{t:1} in Section
\ref{sec:main-result}, where their relevance to the problem of spherical
gravitational collapse is discussed.  Theorem \ref{t:1} is proved in Section
\ref{sec:behav-massl-scal}.

\section{Main results}
\label{sec:main-result}

Consider the exterior region $\D$ of the extreme Reissner-Nordstr\"om black
hole. This region is described in isotropic coordinates $(t,\rho,\theta,\phi)$
by the metric
\begin{equation} \label{metric1}
g = -N^{-2} dt^2 +N^{2} \left(d \rho^2+ \rho^2 (d\theta ^2 + \sin ^2 \theta d
  \phi^2) \right), \;\; N=1+\frac{m}{\rho},
\end{equation}
where the positive constant $m$ represents the total mass of
the spacetime, which equals  the absolute value of the total electric
charge. The electromagnetic field
\begin{equation} \label{emf}
{\cal F} =\pm \frac{m}{(m+\rho)^2} \; dt \wedge d\rho,
\end{equation}
together with this metric, solves the Einstein-Maxwell equations. Note that the
isotropic coordinate $\rho$ differs from the standard radial coordinate
\begin{equation}
r=\rho+m, 
\end{equation}
that gives the area $A= 4 \pi r^2$ of the spheres spanned by acting on a point
with the $SO(3)$ isometry subgroup. Instead of $\rho$, it is often more
convenient to use the ``tortoise´´ radial variable $-\infty < s < \infty$
defined by
\begin{equation}
  \label{defs}
  \frac{ds}{d\r}=N(\r)^2.
\end{equation}
We choose the integration constant such that
\begin{equation} \label{s}
s = \r -\frac{m^2}{\r}+  2m\log\left(\frac{\r}{m}\right). 
\end{equation}
Isotropic coordinates cover only the open exterior region $\D$ of the black
hole.  Figure \ref{f:1} exhibits the well known conformal diagram of an extreme
Reissner-Nordstr\"om black hole (see \cite{Hawking73}, \cite{Carter73} and
references therein), region $\D$ appears shaded. The unshaded region is the
black hole interior, proved to be linearly unstable in \cite{Dotti:2010uc}.
$S$ is a generic $t=$ constant surface, it is a Cauchy surface for $\D$ and a
complete Riemannian manifold with topology $\sph \times \mathbb{R}$.  Its
induced metric approaches that of the cylinder as $\rho \to 0^+$, limit in
which the area of the isometry spheres tend to $A= 4 \pi m^2$.  We denote by
$i_+$ and $i_-$ the future and past timelike infinity of $\D$ respectively. The
asymptotically flat spacelike infinity is denoted by $i_0$, and the
asymptotically cylindrical end of $t=$ constant surfaces is denoted $i_c$.
Note that the surface $S$, being orthogonal to the Killing vector $\p / \p t$,
is asymptotically null at $i_c$.\\

\begin{figure}
\centerline{\includegraphics{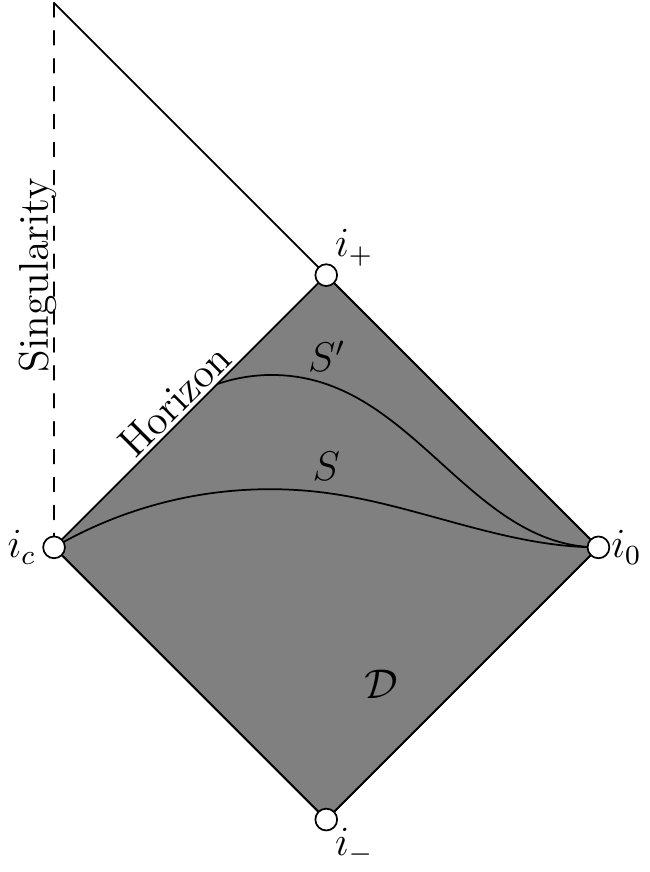}}
\caption{Conformal diagram for the extreme Reissner-Nordstr\"om black hole}
\label{f:1}
\end{figure}

In this work we study the scalar  wave equation on $\D$
\begin{equation}
  \label{laplacian}
 \Box_g \Phi=0,
\end{equation}
with initial data on $S$
\begin{equation} \label{id}
  \phi=\Phi|_S, \quad \chi = \dot \Phi|_S,
\end{equation}
where the dot denotes derivative with respect to $t$.  The existence and
uniqueness of the solution of the Cauchy problem (\ref{laplacian})-(\ref{id})
on a curved background is well established (see, for example, \cite{Hawking73},
also \cite{Friedlander}).
We prove  the following
\begin{theorem}
\label{t:1}
Let $\Phi$ be a solution of the wave equation (\ref{laplacian}) on the open exterior
region $\D$ of an extreme Reissner-Nordstr\"om black hole, which has smooth
initial data (\ref{id}) of compact support on the Cauchy surface $S$. Then,
there exists a constant $C$, which depends only on the initial data (\ref{id}),
such that, in $\D$,
  \begin{equation}
    \label{boundphi}
    |\Phi|\leq \frac{C}{\rho+m}.
  \end{equation}
  All higher partial derivatives with respect to the coordinates $(t,s, \theta,
  \phi)$ are similarly bounded in $\D$. Namely, for any $ \alpha_1 \alpha_2 ...$ there
  exists a constant $C_{\alpha_1 \alpha_2 ....}$ that depends on the initial
  data, such that
\begin{equation}
  \label{T1bound}
  |\partial_{\alpha_1} \partial_{\alpha_2} ... \Phi |\leq \frac{C_{\alpha_1 \alpha_2...}}{\rho+m}.
\end{equation}
\end{theorem}

Theorem \ref{t:1} establishes that the spacetime $\D$ is stable with respect to
this class of initial data.  A key simplification in the study of extreme black
holes, compared to non-extreme ones, is the existence of complete exterior
Cauchy surfaces such as $S$. This has been extensively used in the present work
because it allows us to avoid the delicate issues related to the behaviour of
fields near the horizon.  We should stress that theorem \ref{t:1} applies to
the evolution of data with compact support on $S$.
%
%
By finite speed propagation, it is clear that the restriction of
$\Phi$ to any $t=$ constant slice will have compact support, and, by
smoothness, it will be bounded.  The non-trivial statement in Theorem \ref{t:1} is that
there exists a $t$-independent bound for $\Phi$.


The scalar wave equation on the exterior of an extreme Reissner-Nordstr\"om
black hole has recently been treated by Aretakis in \cite{Aretakis:2010gd},
\cite{Aretakis:2011hc} and \cite{Aretakis:2011ha} by evolving data from a
non-Cauchy surface such as $S'$ in Figure \ref{f:1}.  This surface intersects
the horizon on one end, goes to spacelike infinity on the other end, and is
suitable to analyze the stability of the exterior region of a spherical
collapse of an extreme charged body (see Figure \ref{f:2}).  Among the many
estimates found in these articles, two of them relate closely to our result.
The first one is the pointwise boundedness of the solution in the domain of
dependence of $S'$, in terms of initial data on $S'$, given in Theorem 4 in
\cite{Aretakis:2011hc}.  A weaker bound than  (\ref{boundphi}), $|\Phi|<C$, follows from
this theorem and the fact that, for the kind of initial data we consider,
$\Phi$ has compact support $K$ in the region between $S$ and $S'$.
 There are, however, two important motivations to study
compact data fields.
The first one is that the proof of Theorem \ref{t:1} is considerably simpler
than the boundedness proof  in \cite{Aretakis:2011ha}, since it uses uses only arguments involving the canonical conserved
energy of the wave equation, this being possible due to the existence of the 
above mentioned complete  Cauchy surface.  The cost of this simplification is that  
the results in \cite{Aretakis:2011hc}  apply to a wider class of data on $S'$ than those coming 
from the evolution of fields with compact support on $S$, as we
discuss below.

The second motivation is determining if the blow up result for second and higher
transverse derivatives at the horizon along the horizon generators, reported
in Theorem 6 in \cite{Aretakis:2011hc}, and  generalized to scalar wave equations 
on spacetimes with degenerate Killing 
horizons in \cite{Aretakis:2012ei}, holds for compact data.
To state this result we  need to switch to a  coordinate system that covers the horizon:
\begin{equation}
\{ v=t +s (\rho), r=\rho+m, \theta,\phi \}.
\end{equation}
In these coordinates the metric (\ref{metric1}) has the form
\begin{equation}
g = -(1-m/r)^2 \; dv^2 + 2 dv dr + r^2 \; (d \theta^2 + \sin^2 \theta \; d \varphi^2), 
\end{equation}
and the entire diagram in Figure \ref{f:1} is covered. The horizon is located
at $r=m$, and region ${\cal D}$ corresponds to $r>m$. The Killing vector field
$\p/\p v$ becomes null at $r=m$, its integral lines are the horizon generators,
and $v$ is an affine parameter. Note that $\p/\p r|_{r=m}$ is a null vector,
orthogonal to the $SO(3)$ orbits, and has unit inner product with $\p/\p v$.
We call $ \p/\p r |_{r=m}$ a {\em transverse derivative at the horizon.}
Theorem 6 in \cite{Aretakis:2011hc} states that second and higher order
transverse derivatives of a solution of the scalar field equation diverge as $v
\to \infty$ along horizon generators.  Although this result holds for generic
data in the class analyzed therein, it fails for those fields which evolve from
data with compact support on $S$, as we discuss in detail in Section
\ref{sec:transv-deriv-at}.  For fields evolving from compact data, we could not
rule out these divergences, although we showed that, if present, they would be
milder than those reported by Aretakis.

In what follows we analyze the behaviour of fields with initial data with
compact support on $S$ of Figure \ref{f:1}.  Consider the spherically symmetric
collapse of charged matter. It is possible to arrange the matter model in such
a way that the exterior region is a portion of an extreme Reissner-Nordstr\"om
black hole.  One way of constructing such spacetime is by using a thin shell of
matter, this has been extensively studied in \cite{boulware73}. Other models
are those of charged dust collapse (see \cite{0264-9381-7-6-008} and references
therein). For the present discussion it is enough to know that such a
construction is feasible with some matter model. We show schematically the
conformal diagram of a charged collapse spacetime in Figure \ref{f:2}, where we
have avoided drawing the singularity, which may have a complicated structure,
irrelevant to our purposes.
\begin{figure}[h]
\centerline{\includegraphics[height=8cm]{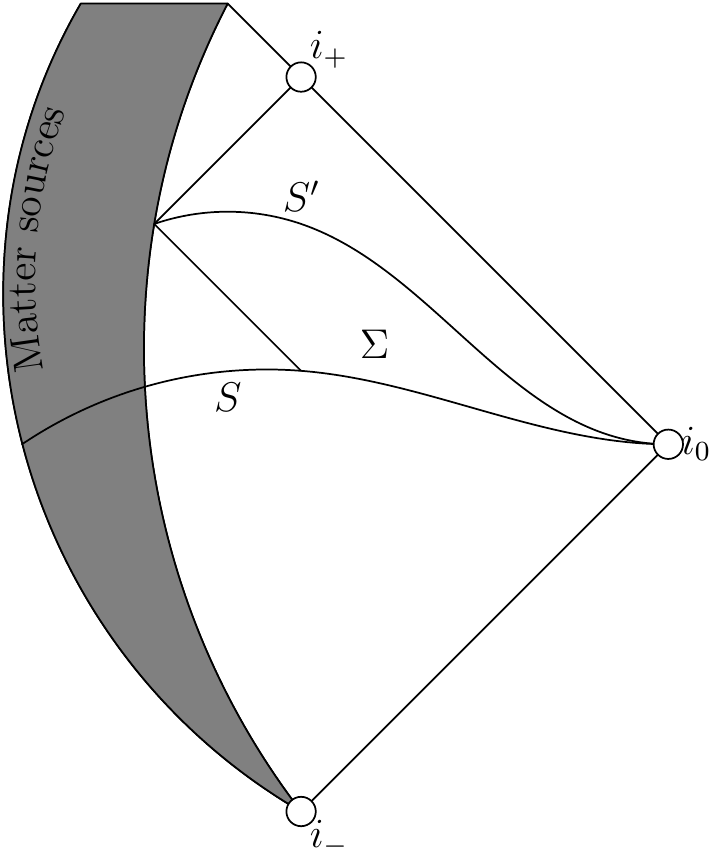}} 
\caption{Conformal diagram for the collapse of charged dust.}
\label{f:2}
\end{figure}
In a realistic collapse of matter there is always a surface like $S'$, for
which the intersection ${\cal I}$ of its future domain of dependence with the
domain of outer communications agrees with that of $S'$ in Figure \ref{f:1}.
Note that the results in \cite{Aretakis:2011ha,Aretakis:2011hc} apply to ${\cal I}$ for 
data  given on $S'$ which may be non trivial at the horizon.  In
Figure \ref{f:2}, $S$ is a Cauchy surface of $\Rt$ topology
that enters the matter region, and $\Sigma \subset S$ is the minimal subset of $S$ whose future contains $S'$. 
The results  \cite{Aretakis:2011ha,Aretakis:2011hc} imply that generic data of compact support within
the vacuum region of $S$ will evolve in a bounded wave in the vacuum portion of this spacetime outside the horizon.
This is so because the portion of this region lying between $S$ and $S'$ is compact and, as the data 
of this field on $S'$ is in the Aretakis class, the field is also bounded in the outer region in the future of $S'$.\\
Our proof of boundedness is simpler, but can only be applied to the spherical collapse model if we restrict to 
 those fields with data of compact support in  $\Sigma \subset S$ (instead of $S$),
for which the evolution in the outer region is identical to that of a field in the Reissner-Nordstr\"om geometry.  
Physically, these fields are characterized by the fact that they reach the horizon after the surface of the collapsing star has crossed it.
We prove below that the growth of these fields along horizon generators is milder than the growth of those fields which enter 
the matter region earlier.
The fact that fields initially supported in
$\Sigma$ are better behaved than those entering the matter before the horizon
is formed is an aspect of the stability of the spherical collapse worth
pointing out.

\section{Behaviour of massless scalar fields} 
\label{sec:behav-massl-scal}

\subsection{Recasting the wave operator}
\label{recast}
Consider the wave equation (\ref{laplacian}) on the exterior ${\cal M}$ of the
extreme Reissner-Nordstr\"om background, whose metric in isotropic coordinates
is (\ref{metric1}).  Defining
\begin{equation}
  \label{defF}
  \Phi = \frac{F}{(\r+m)}, 
\end{equation}
we obtain  
\begin{equation} \label{box2}
-\tfrac{\r^2}{\r+m} \; \Box_g \Phi = \ddot F +\;  \A \; F,
\end{equation}
where a dot means $\p_t$,
\begin{equation} \label{A}
\A = -\p_s^2 + \left( \frac{2m \r^3}{(\r+m)^6}   - \frac{\r^2}{(\r+m)^4} \Delta \right)
 =: -\p_s^2 + \left( V_1 - V_2 \; \Delta \right)
\end{equation}
 $s$ is the tortoise radial coordinate introduced in (\ref{s}),
and $\Lsp$ is the standard Laplacian on the unit sphere. 
From (\ref{s}) we deduce 
\begin{equation}
  s \sim -\frac{m^2}{\r} \; \text{ as } \r \to 0^+ , \;\; s \sim \r  \; \text{
    as } \r \to \infty. 
\end{equation}
The potentials  $V_1$ and $V_2$ are positive, bounded
\begin{equation} \label{Vbounds}
0 < V_1 <  \frac{1}{32 \, m^2} , \quad 0 < V_2 < \frac{1}{16 \, m^2},
\end{equation}
and have the following fall off 
\begin{equation} 
V_2 \sim  s^{-2}, \quad  V_1 \sim 2m  |s|^{-3} \text{ as }  s \to \pm \infty.
\end{equation}
The symmetry in the asymptotic expressions above is not coincidental, $V$ is an
even function of $s$. The origin of this symmetry is the conformal isometry $C$
of extreme Reissner-Nordstr\"om noticed in \cite{couch84}, given by
\begin{equation} \label{ci}
C(t,\r,\theta,\phi)=(t,m^2/\r,\theta,\phi),
\end{equation}
Under this map the pullback of the metric and electromagnetic fields are
\begin{equation}
\tilde g_{ab} = \left( \tfrac{m}{\r} \right)^2  g_{ab},   \quad \tilde {\cal F}_{ab} = -
{\cal F}_{ab}.
\end{equation}

  Since the equation $\Box \Phi - (R/6) \Phi$ is conformally invariant
with conformal weight minus one, and the Ricci scalar of an electro-vacuum
solution vanishes (thus $R = \tilde R =0$), it follows that, if
$\Phi(t,\r,\theta,\phi)$ is a solution of $\Box \Phi =0$, then so is
\begin{equation} \label{fip}
'\Phi(t,\r,\theta,\phi) = (m/\r) \Phi(t,m^2/\r,\theta,\phi).
\end{equation}
For the above field $'F(t,\r,\theta,\phi) = (\r+m) (m/\r)
\Phi(t,m^2/\r,\theta,\phi) = F(t,m^2/\r,\theta,\phi)$. The conformal isometry
is easily expressed in the alternative radial coordinate $s$: under $\r \to
m^2/\r$, $s \to -s$.  The fact that for any solution $F(t,s,\theta,\phi)$ of $
(-\p_s^2 + \left( V_1 - V_2 \; \Delta \right) F =0$, the function
$'F(t,s,\theta,\phi)= F(t,-s,\theta,\phi)$ is also a solution of this equation,
implies that $V_i(s)=V_i(-s), i=1,2$.\\

Note the consistency of the bound (\ref{boundphi}) with the conformal symmetry:
if $\Phi$ has compact support on a $t-$slice, then so does $'\Phi$ given in
(\ref{fip}), however the bound on $'\Phi$ gives no additional information, as
it follows from the bound on $\Phi$:
\begin{equation}
| '\Phi(t,\rho,\theta,\phi)| = \left| \frac{m}{\rho} \Phi(t,m
  2/\rho,\theta,\phi) \right| \leq \frac{m}{\rho} \; \frac{C}{m+(m^2/\rho)} =
\frac{C}{m+\rho}. 
\end{equation}

Finally, although we will not make use of this fact, we  note that the operator
\begin{equation} \label{boxi}
   \p_s^2  + V(s)  \Lsp,
\end{equation}
is the Laplacian on the cylinder ${\cal C} = \mathbb{S}^2 \times {\mathbb R}$ with
respect to the metric
\begin{equation}
  h=V^2 ds^2 + V \;(d\theta ^2 + \sin ^2 \theta d \varphi^2), 
\end{equation}
which reduces to the standard metric  on  ${\cal C}$ if we set $V(s)=1$.
The principal part of the operator $\A$ in (\ref{A}) is a particular case of
(\ref{boxi}).

\subsection{Estimates for functions defined on the cylinder}
\label{efc}

A $t-$slice $S$ of the extreme Reissner-Nordstr\"om spacetime is a cylinder
with the non standard metric induced from (\ref{metric1}). In this Section we
establish pointwise bounds on functions on $S$ from $L^2$ norms defined using
the standard metric on $\sph \times {\mathbb R}$, i.e., we use the hermitian
product
\begin{equation} 
\label{hp}
\langle f , g \rangle = \int f^* \; g \; dx \; \sin(\theta) \, d\theta \, d \varphi,
\end{equation}
($x = s/m$, where $s$ is defined in equation (\ref{s})) and the associated norm
\begin{equation}
||f || = \sqrt{ \langle f , f \rangle }.
\end{equation}
The following result from \cite{Dimock:1987hi} is used in \cite{Kay:1987ax}.
Since we will make use of it, we give a detailed proof using elementary
methods.  
\begin{lemma}
\label{l:boundfcylinder}
 Let $f$ be a complex function on the cylinder with finite norm, then
\begin{equation}\label{sb}
|f(s,\theta,\phi)| \leq M \; \left( ||f|| + m^2 \, || \partial_s^2 \, f || + ||
  \triangle \, f || \right), 
\end{equation}
where $M$ is the constant defined in (\ref{M}).
\end{lemma}

\begin{proof}
A function $f$ of finite norm can be expanded, using spherical harmonics in
  $\sph$ and Fourier transform in ${\mathbb R}$, as
\begin{equation} 
\label{rep}
f(x,\theta,\phi) = \frac{1}{\sqrt{2\pi}} \int dk \sum_{\ell \, m} \hat f_{\ell
  \, m}(k)\; e^{ikx}\; Y_{\ell \, m}(\theta,\phi). 
\end{equation}
where
\begin{equation} \label{rep2}
\hat f_{\ell \, m}(k) := \frac{1}{\sqrt{2\pi}} \int dk f(x,\theta,\phi)\;
e^{-ikx}\; Y^*_{\ell \, m}(\theta,\phi). 
\end{equation}
From equation (\ref{rep2}) we deduce
\begin{align}
 \nonumber
 |f(x,\theta,\phi)| & \leq  \frac{1}{\sqrt{2\pi}} \int dk \sum_{\ell \, m} |\hat f_{\ell
  \, m}(k)|\; |Y_{\ell \, m}(\theta,\phi)|\\
&=\frac{1}{\sqrt{2\pi}} \int dk \sum_{\ell \, m} 
\left[|\hat f_{\ell \, m}(k)|(1+k^2+\ell(\ell+1)) \right] \;
\left[\frac{|Y_{\ell \, m}(\theta,\phi)|}{(1+k^2+\ell(\ell+1)) }  \right] \label{eq:1b}
\end{align}
where in the last line we have just multiplied and divided 
$(1+k^2+\ell(\ell+1))$.
Using  the Cauchy-Schwarz inequality for series
 in (\ref{eq:1b}), then  the Cauchy-Schwarz inequality for integrals
yields
\begin{align}
|f(x,\theta,\phi)| &\leq  \frac{1}{\sqrt{2\pi}} \int dk \sqrt{\sum_{\ell \, m} 
|\hat f_{\ell \, m}(k)|^2(1+k^2+\ell(\ell+1))^2  } \;
 \sqrt{\sum_{\ell \, m} \frac{|Y_{\ell \,
         m}(\theta,\phi)|^2}{(1+k^2+\ell(\ell+1))^2 }   } \nonumber  \\
&\leq \;\sqrt{ \sum_{\ell \, m} \int \frac{| Y_{\ell \, m} (\theta,\phi) |^2}{
    (1+k^2+\ell (\ell+1))^2} \, \frac{dk}{2\pi}} 
\; \sqrt{ \sum_{\ell' \, m'} \int |\hat f_{\ell' \, m'}(k')|^2 \, (1+k'^2+\ell'
  (\ell'+1))^2 \; dk' } \label{bound1}.
\end{align}

The second factor in (\ref{bound1})  can be bounded using
\begin{align} \label{rep3}
\langle \partial_x^2 f , \partial_x^2 f \rangle &= \int \; dk \; \sum_{\ell \,
  m} \; |k^2 \;\hat f_{\ell \, m}(k)|^2,\\ 
\langle \triangle f , \triangle f \rangle &= \int \; dk \; \sum_{\ell \, m} \;
|\ell (\ell+1) \;\hat f_{\ell \, m}(k)|^2, 
\label{rep4}
\end{align}
together with
$(a+b+c)^2 \leq 3 (a^2+b^2+c^2)$ and $\sqrt{a^2+b^2+c^2} \leq |a|+|b|+|c|$,
\begin{multline}
\sqrt{ \sum_{\ell' \, m'} \int |\hat f_{\ell' \, m'}(k')|^2 \, (1+k'^2+\ell'
  (\ell'+1))^2 \;dk' } \\ 
\leq \sqrt{ \sum_{\ell' \, m'} \int 3|\hat f_{\ell' \, m'}(k')|^2 \,
  (1+k'^4+\ell'^2 (\ell'+1)^2)\; dk' } 
=
\sqrt{3} \sqrt{ ||f||^2 + || \partial_x^2 \, f ||^2 + || \triangle \, f ||^2 } \\
\leq \sqrt{3} \left( ||f|| + || \partial_x^2 \, f || + || \triangle \, f ||
\right)
\end{multline}
The identity $\sum_m | Y_{\ell \, m} (\theta,\phi) |^2 = (2 \ell +1 )/(4 \pi)$
in the first factor of (\ref{bound1}) then gives, after restoring units
($s=mx$),
\begin{equation}\label{sb1}
|f(s,\theta,\phi)| \leq M \; \left( ||f|| + m^2 \, || \partial_s^2 \, f || + ||
  \triangle \, f || \right), 
\end{equation}
where
\begin{align}
M^2 &= \frac{3}{8 \pi^2} \; \int \sum_{\ell} \frac{2 \ell + 1}{ (1+k^2+\ell
  (\ell+1))^2} \, dk \\ 
&= -\frac{3}{8 \pi^2} \; \frac{\partial}{\partial \ell} \sum_{\ell} \int
(1+k^2+\ell (\ell+1))^{-1} \, dk \\ 
&= \frac{3}{16 \; \pi} \; \sum_{\ell} \frac{2\ell+1}{(1+ \ell(\ell+1))^{3/2}} \label{M}
\end{align}

\end{proof}

Applying Lemma \ref{l:boundfcylinder} to
$F_t(s,\theta,\phi):=F(t,s,\theta,\phi)$ we get a bound for $|F|$ {\em on the
  t-slice}. However, if the terms $||F_t||, ||\Delta F_t||$ and
$|| \partial_s^2 \, F_t ||$ on the right hand side of (\ref{sb}) were further
bounded by {\em conserved} ($t$-independent) slice integrals, then we would get
a $t$-independent bound for $|F_t|$, i.e., a bound for $|F|$. This is our
motivation to study conserved energies in the following Section.

\subsection{Conserved energies}
\label{sec:conserved-energies}

In Section \ref{recast} we have  shown that  the problem (\ref{laplacian})-(\ref{id})
of propagation of scalar waves on the exterior extreme Reissner-Nordstr\"om 
spacetime ${\cal M}$ can be reformulated as the equation
\begin{equation} \label{eqF}
{\cal O} F := \ddot F +\;  \A \; F=0,
\end{equation}
with initial data $(f,g)$ of compact support on the $t=0$ slice
$S=\mathbb{S}^2 \times \mathbb{R}$, given by
\begin{equation}
  \label{idf}
  f=F(t=0)=\phi/(\r+m), \quad  g=\dot F(t=0)=\chi/(\r+m). 
\end{equation}
A solution of equation (\ref{eqF}) has a conserved (i.e., $t-$independent) energy
\begin{equation}
  \label{ceF}
  \E [F]= \int_{S_t} ( \dot F^2 + F\A (F) ) \; \dv,
\end{equation}
where the integral is performed on a $t-$slice  using the standard volume
element $dv= \sin(\theta) \;  d\theta\;  d \varphi \; ds$. This energy is useful
because  it provides a $t-$independent bound for $|| \p_s F ||$.\\

Taking derivatives with respect to $t$ to equation (\ref{eqF}) shows that 
$\dot F, \ddot F,...$ all satisfy the same equation, giving extra conserved quantities,
in particular
\begin{equation}
  \label{em1a}
  \E [\dot F]= \int_{\C} (\ddot F^2 + \dot F\A (\dot F) ) \; dv, 
\end{equation}
which, using  (\ref{eqF}) to substitute for $\ddot F$ reduces to 
\begin{equation}
  \label{ep1}
   \E [\dot F]= \int_{\C} [(\A F)^2 + \dot F\A (\dot F)]\;  \dv.
\end{equation}
The above conserved energy is useful because it provides a $t-$independent
bound for $|| \A F ||$, and thus for $|| \partial_s^2 F ||$. We would like to
get similar bounds for $|| (\Delta F) ||$ and $||F||$, and use them in
(\ref{sb}). Following \cite{Dafermos:2008en}, equation (\ref{ceF}) suggests
that, to obtain a bound for the integral of $F^2$, we should consider the
energy of a ``time integral'' $\tilde F$ of the solution $F$ in (\ref{eqF}), i.e.  a solution of the system
\begin{equation}
  \label{tF}
  \dot{\tilde F}=F  \; \text{ and } \; \ddot {\tilde F} +\;  \A \; \tilde F=0.  
\end{equation}
Assume there exists such a solution, its  conserved energy would be 
\begin{align}
   \E [\tilde F] &= \int_{\C} \left[(\dot {\tilde F})^2 + \tilde  F\A( \tilde
     F) \right] \dv,\\ 
 &=  \int_{\C} \left[ F^2 + \tilde  F\A( \tilde F) \right] \dv ,  \label{em1}
\end{align}
and would bound $||F||$. Using $[\Delta, {\cal O}]=0$ in (\ref{tF}) we could
prove that ${\cal O} \Delta \tilde F=0$, the conserved energy of this solution
being
\begin{align}
  \E [\Delta \tilde F] &= \int_{\C} \left (\Delta \dot{\tilde F}) ^2 + \Delta
    \tilde F\A(
    \Delta \tilde F) \right] \dv,\\
  &= \int_{\C} \left (\Delta F) ^2 + \Delta \tilde F\A( \Delta \tilde F)
  \right] \dv , \label{edelF}
\end{align}
which bounds $||\Delta F||$.\\
In order to proceed with this idea, we have to prove that a solution of (\ref{tF}) exists for  $F$ satisfying 
(\ref{eqF})-(\ref{idf}). This is done in the following Section.

\subsection{Integrating in time}
\label{sec:integrating-time}

In this section we will prove the existence of the ``time integral'' solution 
$\tilde F$ of the system (\ref{tF}) for a solution $F$ of (\ref{eqF})-(\ref{idf}).
The existence of $\tilde F$ implies the conservation of the  energies (\ref{em1}) and 
(\ref{edelF}) for $F$. The proof requires a notion of the inverse $\A ^{-1}$.
This is introduced by taking advantage of the fact that ${\cal A}$ is positive
definite  (since $V_1$ and $V_2$ are positive), and this allows us  
 to define an inner product on functions on the linear space $\hat {\cal H}$ 
of functions of compact support on $S$, 
\begin{equation}
  (p,q)=\int_{S} [  \partial_s p \partial_s q + V_2 (\partial_\theta
  p \partial_\theta q + (\sin\theta)^{-1}  \partial_\phi
  p \partial_\phi q ) +V_1 pq ]  \dv,
\end{equation}
which is formally obtained by integrating by parts $\int_{S} p \A q \; dv$.
We define a Hilbert space $\Hi$ as the completion of $\hat {\cal H}$
under this norm.
This is the space where $\A^{-1}$ is defined, as shown in the following.

\begin{lemma}
\label{l:inverseA}
Let $q$ be smooth and having  compact support $S_q \subset S$. Then, there exist a
unique solution $p \in \Hi$ of the equation
  \begin{equation}
    \label{invA}
    \A(p)=q.
  \end{equation}
Moreover, $p$ is smooth.
\end{lemma}

\begin{proof}
  The Lax-Milgram theorem (see for example \cite{Gilbarg}) asserts that 
if $B(\cdot, \cdot)$ is a bilinear form on $\Hi$ for which  there exists $\alpha, \beta>0$ 
such that 
\begin{equation} \label{lm1}
|B(u,v)| \leq \alpha \sqrt{(u,u)} \;  \sqrt{(v,v)}, \;\; (u,u) \leq \beta \; B(u,u),
\end{equation}
hold for all $u,v \in \Hi$ then, for any bounded linear operator $L: \Hi \to
{\mathbb R}$, the equation on $p$
\begin{equation} \label{lm2}
B(z,p)=L(z) \;\; \text{ for every }  \; z \in \Hi, 
\end{equation}
has a unique solution.\\
Conditions (\ref{lm1}) are trivially satisfied for $\alpha=\beta=1$ if
$B(p,q)=(p,q)$ (Schwarz's inequality).

The operator $L(z):=\int_{S} q z \; dv$ can be shown to be bounded by using the
fact that the restriction of $V_1$ to $S_q$ has a minimum $V_1^{(q)} >0$ and
thus, applying Schwarz's inequality to $L^2[{\cal S}_q,dv]$ gives
\begin{multline}
| L(z) | = \left| \int_{{ S}_q } q z \; dv \right| \leq \sqrt{ \int_{{S}_q} q^2 \; dv} \; \sqrt{ \int_{{S}_q} z^2 \; dv} \\ <
\sqrt{ \int_{{S}_q} q^2 \;dv} \; \sqrt{ \int_{{S}_q} \left( \frac{V_1}{V_1^{(q)}} \right) z^2 \; dv} 
<   \sqrt{ \int_{{S}_q}  \left( \frac{q^2}{V_1^{(q)}} \right) \;  dv}\;
\sqrt{(z,z)} 
\end{multline}
It follows from the Lax-Milgram theorem that there exists a $p \in \Hi$ such that 
\begin{equation}
\int_{S} z \A p \; dv = \int_{S} z q \; dv, \; \; \text{ for every }  \; z \in \Hi,
\end{equation}
and hence $p$ is a weak solution of the elliptic differential equation $\A p =
q$.  Smoothness of $p$ follows from interior elliptic regularity arguments (see
\cite{Gilbarg}) and thus $p$ satisfies $\A p=q$. 

\end{proof}

An alternative proof using expansions in spherical harmonics $Y_{\ell
  m}(\theta,\phi)$ sheds light on the behaviour of $p$ at infinity and near the
horizon.  Introducing $P=p/N$, equation (\ref{invA}) reads
\begin{equation} \label{ode}
\A(p)=\A(NP) = N^{-3} \left( \p^2_{\rho} P + r^{-2} \triangle P \right) = q.
\end{equation}
If we expand $P = \sum_{\ell m} P_{\ell m}(\rho) Y_{\ell m}(\theta,\phi)$ (and similarly expand $q$) 
the above equation reduces to an ODE for each  mode that can  solved  explicitly:
\begin{equation} \label{solode}
(2\ell +1)  \, P_{\ell m}(\rho) = \rho^{\ell+1}  \int^{\rho}_{ C_{\ell m}} x^{-\ell} N^3(x) q_{\ell m}(x) \, dx
-  \rho^{-\ell}  \int^{\rho}_{ D_{\ell m}} x^{\ell+1} N^3(x) q_{\ell m}(x) \, dx.
\end{equation}
For every harmonic mode, $ C_{\ell m}$ and $D_{\ell m}$ are constants of
integration of the generic solution above, however if we use the fact that the
$q_{\ell m}$ have compact support, we conclude that the choice $ C_{\ell m} =
\infty, D_{\ell m}=0$ is the only one that gives an appropriate asymptotic
behavior near the horizon and spatial infinity, such that $p$ belongs to
$\Hi$. This gives (after multiplication times $N$) the unique $p$ singled out
by the Lax-Milgram theorem. Its projection
\begin{equation} \label{pl}
 p_{\ell} := \sum_{m=-\ell}^{\ell} p_{\ell m}(\rho) Y_{\ell m}(\theta,\phi),
\end{equation} 
onto the $\ell$ subspace behaves as 
\begin{equation} \label{beh}
 p_{\ell} \sim \begin{cases} \rho^{\ell} & \text{ as }\rho \to  0^+, \\
 \rho^{-\ell} & \text{ as }\rho \to \infty.
\end{cases}
\end{equation}
Thus, generically, $p$ approaches a constant in both the $\rho \to \infty$ and the $\rho \to 0^+$ limits.

Lemma \ref{l:inverseA} refer to the space variables on the cylinder $S$,
however the functions involved in the proof of existence of equations
(\ref{tF}) depend also on the time parameter $t$. The following remarks, which
concern the $t$ dependence of the functions, are useful in the proofs. Let
$q(t,s,\theta,\phi)$ be a smooth function on $\mathbb{R}\times S$ which has
compact support on $S$ for every $t$.  Note that this is our appropriate class
of functions since they arise as solutions of the wave equation with compact
support, smooth, initial data.  Let $p$ be the solution of
\begin{equation}
  \label{eq:1}
\A(p)=q.  
\end{equation}
From Lemma \ref{l:inverseA} we deduce that $p$ is smooth on $S$. To obtain
smoothness with respect to $t$ we take $t$ derivatives to equation
(\ref{eq:1}). The function $\dot p $, if it it exists, should satisfy the
equation
\begin{equation}
  \label{eq:2}
 \A(\dot p )= \dot q.
\end{equation}
However, by hypothesis $\dot q$ is smooth and has compact support on $S$, hence
we can use Lemma \ref{l:inverseA} to prove that $\dot p$ exists and it is
smooth on $S$. Taking an arbitrary number of $t$ derivatives we conclude that
$p(t,s,\theta,\phi)$ is smooth on $\mathbb{R}\times S$.

Partial derivatives with respect to $t$ clearly conmute with $\A$, to prove that
they also conmute with $\A^{-1}$ for this class of functions we 
 write equation  (\ref{eq:1}) and  (\ref{eq:2})  as
\begin{align}
  \label{eq:3}
 p &= \A^{-1} (q),\\
  \dot p &= \A^{-1} (\dot q). \label{eq:3b}
\end{align}
Taking a time derivative of (\ref{eq:3}) and using equation (\ref{eq:3b}) we obtain
\begin{equation}
  \label{eq:4}
 \frac{\partial }{\partial t} \A^{-1} (q)=\A^{-1} (\dot q).
\end{equation}

We have all the ingredients to prove that there is a solution to equations
(\ref{tF}). We emphasize that in the following proof we do not make use of the
decay behaviour (\ref{beh}), we only need the statement of Lemma
\ref{l:inverseA}.

\begin{lemma}
  \label{l:int-time}
For a given solution $F$ of equation  (\ref{eqF}) with initial data
(\ref{idf}) of compact support,  there exist a  solution $\tilde F$
of equations (\ref{tF}), and the energies (\ref{em1}) and (\ref{edelF}) are
finite and conserved.     
\end{lemma}

\begin{proof}
Consider the function
\begin{equation} \label{ttF}
\tilde{\tilde F} = - \A^{-1} F.
\end{equation}
The function $\tilde{\tilde F}$ exists and it is smooth by Lemma
\ref{l:inverseA}, since $F$ and all its time derivatives have compact support in $S$  for all $t$. 
Note that
\begin{equation}
\ddot{\tilde{\tilde F}} = -\A^{-1} \ddot F = \A^{-1} \A F = F = -\A \tilde{\tilde F}.
\end{equation}
This equation  shows that $\tilde{\tilde F}$ is a solution of the wave equation and a second time integral 
of $F$, i.e.,  $\ddot{\tilde{\tilde F}}=F$.
This immediately implies that 
\begin{equation}
\tilde F = \dot{\tilde{\tilde F}},
\end{equation}
is also a solution of the wave equation, and a first  time integral of $F$, i.e., $\dot{\tilde F} =F$.
This first time integral has finite energy, since 
\begin{equation}
   \E [\tilde F] = \int_{S} \left[ (\dot {\tilde F})^2 + \tilde  F\A( \tilde F) \right] \dv
=  \int_{S} \left[ F^2 + \tilde  F\A( \tilde F) \right]\dv \label{em1b},
\end{equation}
and 
\begin{equation}
\A  \tilde F =    \frac{\partial }{\partial t} \A  \tilde{\tilde F} = -\dot F,
\end{equation}
so both terms in the integrand in (\ref{em1b}) have compact support.
Note, however, that 
\begin{equation}
\label{eq:Etildetilde}
   \E \left[\tilde{\tilde F}\right] = \int_{S} \left[ \left(\dot{\tilde{\tilde
           F}}\right)^2 + \tilde{\tilde{F}}\A \tilde{\tilde{F}} \right] \dv
   =\int_{S} \left[ ( {\tilde F})^2 - \tilde{\tilde  F} F  \right] \dv, 
\end{equation}
diverges as a consequence of the behaviour (\ref{beh}) of $\tilde F = -\A^{-1}
\dot F$ in the first term in the integrand (equation (\ref{beh}) applies to
this case since $\dot F$ has compact support
on $t-$slices.)

The conservation of $ \E [\tilde F] $ in (\ref{em1b}) follows from
\begin{equation}
\frac{d}{dt}  \E [\tilde F] =  \int_{S} \left[ 2 F \dot F  +  F\A  \tilde F + \tilde F \A F  \right]\dv = 
 \int_{S} \left[ 2 F (\ddot{\tilde  F}  +  \A  \tilde F ) \right]\dv = 0.
\end{equation}
All the integrations by parts above are possible since always one of the factors have compact support. 
Note also that 
\begin{equation}
  \label{eq:5}
  \tilde F =\A^{-1} G,
\end{equation}
 where $G$ is the solution of the wave equation (\ref{eqF}) with initial   data with compact 
\begin{equation} \label{dataFt}
G(t=0) = -g, \quad \dot G (t=0) = \A(f).
\end{equation}
  The finiteness and conservation of $\E [\Delta
\tilde F]$ follows from similar arguments using $[\Delta,\A]=0$. 

\end{proof}

It is also possible to prove the existence of $\tilde F$ directly from equation
(\ref{eq:5}) without constructing the second time integral $\tilde{\tilde
  F}$. Namely, take the solution $G$ of the wave equation with data
(\ref{dataFt}). The function $G$ has compact support on $S$ for every $t$,
hence there exist $\tilde F$ such that (\ref{eq:5}) holds. We have chosen to
construct first $\tilde{\tilde F}$ because this function could be useful in
future aplications.  The existence of $\tilde F$ can also be proved 
by acting with $\A$
only on initial data (instead of on a function that depends on $t$ as in
(\ref{eq:5})). That is, consider the solution $\tilde F$ of the wave equation
(\ref{eqF}) with the following initial (see equation (\ref{dataFt})) data
\begin{equation}
  \label{eq:6}
  \tilde F(t=0)=-\A^{-1}(g), \quad \dot{\tilde F}(t=0)=f.
\end{equation}
The initial data have not compact support but the solution of the wave
equation, by finite speed propagation, nevertheless exists and it is smooth.

\subsection{Proof of theorem \ref{t:1}: bound on $\Phi$}

Equation (\ref{laplacian}) subject to (\ref{id}) is equivalent to (\ref{eqF})
subject to (\ref{idf}).  Let $F_t(s,\theta,\phi)=F(t,s,\theta,\phi)$ be the
restriction of $F$ to a $t$ slice. Since the slice is $S^2 \times {\mathbb R}$
and $F_t$ has compact support, the bound (\ref{sb}) holds, then
\begin{equation}
|F_t(s,\theta,\phi)| \leq M \;  \left( ||F_t|| + m^2\; || \partial_s^2 \, F_t || + || \Delta \, F_t || \right).
\end{equation}
On the other hand, from (\ref{A}), 
\begin{equation}
|| \p_s^2 F_t || \leq || \A F_t || + V_1^{max} ||F_t|| + V_2^{max} || \Delta F_t||
\end{equation}
where $V_{1,2}^{max}$ are the maximum of the positive potential $V_{1,2}$, given in (\ref{Vbounds}).
Thus 
\begin{align} \nonumber
|F_t(s,\theta,\phi)| & \leq M \;  \left( m^2 \, || \A F_t ||+ (1+m^2 \, V_1^{max}) ||F_t|| +
 (1+ m^2\, V_2^{max}) || \Delta F_t|| \right) \\
& \leq  \tfrac{17 }{16} M \;  \left( m^2 \, || \A F_t ||+  ||F_t|| + || \Delta F_t|| \right)
\end{align}
However
\begin{equation}
||F_t||  \leq \sqrt{\E [\tilde F]}, \;\;\;  || \A  F_t || \leq \sqrt{\E [\dot F]},\;\;\; 
|| \Delta \, F_t || \leq \sqrt{\E [\Delta \tilde F]}, 
\end{equation}
and the above energies are $t-$independent, thus
\begin{equation}
\label{eq:bound}
|F(t,s,\theta,\phi)| \leq  \tfrac{17 }{16} M\;
\left( m^2 \sqrt{\E [\dot F]} + \sqrt{\E [\Delta \tilde F]} + \sqrt{\E [\tilde F]} \right) =: C,
\end{equation} 
where $C$ is a constant, and  (\ref{boundphi}) follows.

It is important  to note that if we attempt to replace $F$ by $\tilde F$ in the
bound \eqref{eq:bound} (that is, if we want to prove that $\tilde F$ is
bounded) then the last term in the right hand side of \eqref{eq:bound} is
$\sqrt{\E [\tilde{\tilde F}]}$, and we have seen that this energy is not bounded
  (see equation \eqref{eq:Etildetilde}).

\subsection{Proof of Theorem \ref{t:1}: bounds on higher derivatives} 
\label{hd}

The metric (\ref{metric1}) admits a  four dimensional space of  Killing vector fields, the span of 
\begin{align} 
K_1 & = \cos (\phi) \; \p_{\theta} - \cot(\theta) \; \sin(\phi) \p_{\phi} \\
K_2 &=  \sin (\phi) \; \p_{\theta} - \cot(\theta) \; \cos(\phi) \p_{\phi} \\ \label{killings}
K_3 &= \p_{\phi} \\ 
K_4 &= \p_t
\end{align}
Since the wave operator $\Box$ commutes with Lie derivatives along Killing
vector fields, given any solution $\Phi$ of equation (\ref{laplacian}),
$\pounds_{K_{i_1}} \cdots \pounds_{K_{i_j}}\Phi$ will also be a solution of this
equation. Applying (\ref{boundphi}) to this solution we obtain the bound
\begin{equation} \label{bpd1}
| \pounds_{K_{i_1}}.... \pounds_{K_{i_j}}\Phi | \leq \frac{C_{i_1 i_2...i_j}}{\rho+m},
\end{equation}
where $C_{i_1 i_2...i_j}$ is a constant.\\
At every point $p$ of the spacetime ${\cal M}$, the vectors (\ref{killings})
span a 3-dimensional subspace of the tangent space $T_p {\cal M}$, equation
(\ref{bpd1}) fails to give a bound for derivatives of $\Phi$ along radial
directions.  To obtain a pointwise bound for $\p_s F$, $F$ a solution of
(\ref{eqF}), we proceed as follows: equation (\ref{sb}) applied to $\p_s F$
gives
\begin{equation}\label{bFs}
  |\p_s F| \leq M \;  \left( ||\p_s F|| + m^2 \, || \partial_s^3 \, F || + || \triangle \, \p_s F || \right).
\end{equation}
The first term on the right hand side above has a $t-$independent bound given by the $\E [F]$, since 
\begin{equation}
   \E [F]\geq \int_{S} ( F \A F) \; dv \geq \int_{S} |\partial_s F|^2  dv.
\end{equation}
The last term is similarly bounded by the energy $\E [\triangle F]$ (note that
$[{\cal O},\triangle]=0$, then $\triangle F$ is a solution of the field
equation if $F$ is a solution). To bound the second term, we use the fact that
$[{\cal O},\A]=0$, and compute the energy of $\A F$:
\begin{equation}
   \E [\A F] \geq  \int_{S}  (\A F)\A (\A F)  dv \geq \int_{S} |\partial_s\A F|^2.
\end{equation}
The last integrand is (see (\ref{A}))
\begin{equation} \label{sder}
  \partial_s\A F=-\partial^3_s F+F \partial_s V_1 +V_1\partial_s F - (\partial_s 
  V_2) \triangle F-V_2 \triangle \partial_s  F
\end{equation}
and it is easy to check that the functions $\partial_s V_i$ are bounded, 
say $|\partial_s V_i| \leq V_{i,s}^{max}$. Thus
\begin{multline} \nonumber
   ||\partial^3_s F|| \leq || \partial_s\A F || + V_{1,s}^{max} ||F|| + V_{1}^{max} ||\p_s F||
+  V_{2,s}^{max} ||\triangle F|| +  V_{2}^{max}  ||\p_s \triangle F || \\
\leq \E [\ A F] +  V_{1,s}^{max} \E [\tilde F] +  V_{1}^{max} \E [F] +
 V_{2,s}^{max} \E [\triangle \tilde F] +  V_{2}^{max} \E [ \triangle F ]
\end{multline}
which is $t-$independent. We conclude that the right hand side of equation (\ref{bFs}) can be bounded
by a $t-$independent constant.

It is easier to find a pointwise bound for the second radial derivative, since  
using  (\ref{A}),
\begin{equation}
  |\partial^2_sF|\leq |\A F| + V_1^{max} |F| +   V_2^{max} |\triangle F|.
\end{equation}
and both $\A F$ and $\triangle F$ are solutions of (\ref{eqF}), and therefore, pointwise bounded.

To bound $|\p_s ^3 F|$ we may use (\ref{sder})
\begin{equation}
  |\partial^3_s F| \leq | \partial_s\A F | + V_{1,s}^{max} |F| + V_{1}^{max} |\p_s F|
+  V_{2,s}^{max} |\triangle F| +  V_{2}^{max}  |\p_s \triangle F |, 
\end{equation}
together with the fact that every function on the right hand side is either a
solution of (\ref{eqF}) with compact support on $t-$slices, or an
$s-$derivative of such a solution, all of which are bounded.  Higher
$s-$derivatives can be bounded in this way by induction: take $(n-3)$
$s-$derivatives of equation (\ref{sder}), this gives $\p_s^n F$ in terms of
lower $s-$derivatives of $F, \A F$ and $\triangle F$, all of which, being
solutions of (\ref{eqF}) with compact support on $t-$slices, are pointwise
bounded by the inductive hypothesis. These terms come multiplied by higher
$s-$derivatives of the $V_i$, but these can be easily shown to be bounded by
noting that
$\p_s^k V_i$ is a polynomial in $z :=1/(r+m)$.\\

\subsubsection{Transverse derivatives at the horizon}
\label{sec:transv-deriv-at}

In \cite{Aretakis:2011hc}, some transverse derivatives of $\Phi$ across the
horizon were found to diverge along the horizon generators.  In this Section we
show that compact data fields, which belong to a subclass of the solutions
studied in \cite{Aretakis:2011hc}, are better behaved.

Since the coordinates $\{ t, \rho,\theta,\phi\}$ (or $\{ t, s,\theta,\phi\}$)
cover only the exterior region of the black hole, we need to switch to advanced
coordinates $\{ v=t+s,r=\rho+m,\theta,\phi \}$ in order to properly state this
problem.  Note that
\begin{align}
\left. \frac{\p }{\p s}\right|_{\{ t,\theta,\phi \} } &= \left. \frac{\p}{\p
    v}\right|_{\{ r,\theta,\phi \} } + \left( \frac{r-m}{r} \right)^2  
 \left. \frac{\p}{\p r}\right|_{\{ v,\theta,\phi \} },\\
\left. \frac{\p }{\p t}\right|_{\{ s,\theta,\phi \} } &=  \left. \frac{\p}{\p
    v}\right|_{\{ r,\theta,\phi \} }, 
\label{dv}
\end{align}
become linearly dependent at the horizon and that the norm of $\left. \frac{\p
  }{\p s}\right|_{\{ t,\theta,\phi \} }$ vanishes when $r \to m^+$. Thus,
although we have proved the pointwise boundedness of partial derivatives along
the coordinates $\{t,s,\theta,\phi \}$, which are suitable and span the tangent
space at any point outside the horizon, the study of the transverse derivatives
$ \left. \frac{\p}{\p r}\right)_{\{ v,\theta,\phi \} }$ {\em at the horizon}
requires a separate treatment.

In advanced coordinates, (\ref{metric1}) reads
\begin{equation} \label{rnvr} ds^2 = -(1-m/r)^2 \; dv^2 + 2 dv dr + r^2 \; (d
  \theta^2 + \sin^2 \theta \; d \varphi^2),
\end{equation}
and the scalar wave equation is 
\begin{equation} \label{boxvr}
\Box \Phi =  \left( \frac{r-m}{r} \right)^2  \p_r^2 \Phi + 2 \left(
  \frac{r-m}{r^2} \right) \p_r \Phi +2 \p_r \p_v \Phi +\frac{2}{r} \p_v \Phi +
\frac{\triangle}{r^2} \Phi =0.
\end{equation}
 
Theorem 1 in \cite{Aretakis:2011hc} states that for every $\ell$ there exists a
set of constants $\beta_i$ such that the functions $H_{\ell} [\Phi]$, defined
on the horizon $r=m$ as
\begin{equation}
H_{\ell} [\Phi] := \left[\p_r^{\ell +1} \Phi_{\ell}+ \sum_{i=0}^{\ell} \beta_i
  \p_r^{i} \Phi_{\ell} \right]_{r=m}, 
\end{equation}
are constant along the horizon generators, i.e., they depend on $(\theta,\phi)$
but not on $v$.  Here, $ \Phi_{\ell}$ is the projection of $\Phi$ onto the
$2\ell+1$ dimensional $\ell$ harmonic space on $\sph$ (as in equation
(\ref{pl})), which is itself a solution of the wave equation.  This theorem
implies that a generic solution $\Phi$ of the wave equation within the class
studied in \cite{Aretakis:2011hc} does not admit a time integral solution
$\tilde \Phi$ in the sense of (\ref{tF}), as the existence of such a time
integral would imply that $H_{\ell} [\Phi] = \p_v H_{\ell} [\tilde \Phi] \equiv
0$, whereas the $H_{\ell} [\Phi] $ are generically non-trivial for these
fields, as they evolve from data give on a surface crossing the horizon, and
the data are generically  non-trivial at the horizon. The reason why this
result does not contradict Lemma \ref{l:int-time} above lies in the fact that
the solutions of compact data studied here are a subclass of the set studied in
\cite{Aretakis:2011hc}, and the $H_{\ell} [\Phi]$ trivially vanish for this
subclass, as can be seen by taking the limit $v \to -\infty$ along horizon
generators (where eventually $\Phi$ is trivial), and using
the fact that $H_{\ell} [\Phi]$ does not depend on $v$.

Theorem 6 in \cite{Aretakis:2011hc} states that some transverse derivatives at
the horizon blow up along generators, more precisely,
\begin{equation} \label{div}
\partial_r^{\ell+m+k} \p_v^m \Phi_{\ell} \sim v^{k-1}, 
  \text{ as } v \to \infty, \; k \geq 2, m\geq 0 ,
\end{equation}
where the limit is taken along a horizon generator, i.e., $v \to \infty$ while
keeping $r=m$, and $(\theta,\phi)$ fixed.  The proof of this theorem, however,
requires that the $H_{\ell} [\Phi]$ be non trivial, and so this result does not
hold for compact data solutions.

The worst divergences of transverse derivatives along the horizon reported in
\cite{Aretakis:2011hc} come from the $\ell=0$ piece of $\Phi$. Let us then
assume, for simplicity, that $\Phi=\Phi_{\ell=0}$ is a spherically symmetric
solution of (\ref{boxvr}), then the last term in (\ref{boxvr}) vanishes and
\begin{equation}
\p_v (\p_r \Phi + \tfrac{1}{m} \Phi) \= 0,
\end{equation}
 where $ \dot{=}$ means ``equal at the horizon''. 
Integrating this equation, we prove the constancy along the horizon  generators of 
\begin{equation} \label{h0}
H_{\ell=0} [\Phi]=\p_r \Phi + \tfrac{1}{m} \Phi,
\end{equation}
 which is one of the $H_{\ell}$ referred to above. Note that 
this conserved quantity implies the
 boundedness of $\p_r \Phi$ at the horizon. 
If we  take the  $r$-derivative of equation (\ref{boxvr}),
evaluate this equation  at the horizon, and then  use the original
equation  (\ref{boxvr}) (evaluated at the horizon) 
to eliminate the term $\p_v \p_r \Phi$ we obtain 
\begin{align}
\p_v \p_r^2 \Phi &= \tfrac{2}{m^2} \p_v \Phi - \tfrac{1}{m^2} \p_r \Phi,\\
&= \tfrac{2}{m^2} \p_v \Phi - \tfrac{1}{m^2} H_{\ell=0}+\tfrac{1}{m^3}\Phi,
\end{align}
where in the last equality we have used equation \eqref{h0}. 
Integrating this equation in $v$, and using (\ref{h0}) gives
\begin{equation} \label{pr2h}
\p_r^2 \Phi = \p_r^2 \Phi |_{v_0}  + 
\frac{2}{m^2} \left( \Phi - \Phi|_{v_0}  \right) - \frac{H_0}{m^2} (v-v_0) 
+ \frac{1}{m^3} \int_{v_0}^v \Phi \; dv.
\end{equation}
Since $|\Phi| \leq C  v^{-3/5}$ for large $v$ along a horizon generator and some constant $C$ \cite{Aretakis:2011hc},
it follows from \eqref{pr2h} that, if $H_0 \neq 0$,  $\p_r^2 \Phi \sim v$ along generators.
However, for fields with compact data, $H_0=0$ in \eqref{pr2h} implies that
$|\p_r^2 \Phi| \leq C'  v^{2/5}$ in this same limit, $C'$ a constant.
 This is a distinctive feature of compact data
fields.\\

Now suppose that $\tilde \Phi$ belongs to the class studied by Aretakis.  
 Then  we could rewrite \eqref{pr2h} as
\begin{equation} \label{pr2ha}
\p_r^2 \Phi = \p_r^2 \Phi |_{v_0}  + 
\frac{2}{m^2} \left( \Phi - \Phi|_{v_0}  \right)
+ \frac{1}{m^3} \left( \tilde \Phi - \tilde \Phi|_{v_0}  \right), 
\end{equation}
and use boundedness of $| \tilde \Phi|$ to prove boundedness of $\p_r^2 \Phi$
at the horizon. More generally, we could apply (\ref{div}) to $\tilde \Phi$ and
arrive at
\begin{equation} \label{div2}
  \partial_r^{\ell+n+q} \p_v^n \Phi_{\ell} \sim v^{q-3}, 
  \text{ as } v \to \infty, \; (q \geq 4, n \geq 0) .
\end{equation}

We do not have a proof that $\tilde \Phi$ could be extended to a field in the
class of solutions in \cite{Aretakis:2011hc}, as, in principle, $\tilde \Phi$
is only defined in the open set $r>m$. However, the facts that $\tilde \Phi
\sim r^{-1}$ near spacelike infinity (see (\ref{beh})) and $E[\tilde \Phi]<
\infty$ suggest that such an extension exits.  \footnote{We thank S. Aretakis
  and H. Reall for this observation.} Note that it is unlikely that we could
further extend these arguments to $\tilde{\tilde {\Phi}}$, since this field has
divergent energy.

 \section*{Acknowledgments}
 We would like to thank Lars Andersson, Pieter Blue, Mihalis Dafermos and
 Martin Reiris for illuminating discussions. We specially thank Stefanos
 Aretakis and Harvey Reall for pointing out errors in previous versions of this
 manuscript, and making comments that led to many improvements.

 The authors are supported by CONICET (Argentina). This work was supported by
 grants PIP 112-200801-02479 and PIP 112-200801-00754 of CONICET (Argentina),
 Secyt 05/B384 and 30720110101131 from Universidad Nacional de C\'ordoba
 (Argentina), and a Partner Group grant of the Max Planck Institute for
 Gravitational Physics (Germany).


\end{document}